\documentclass[a4paper,UKenglish,cleveref,autoref,thm-restate]{lipics-v2021}
\pdfoutput=1

\title{Approximating the Value of Energy-Parity Objectives in Simple Stochastic Games}
\titlerunning{Approximating the Value of Energy-Parity Games}

\author{Mohan Dantam}{School of Informatics, University of Edinburgh, UK}{}{}{}
\author{Richard Mayr}{School of Informatics, University of Edinburgh, UK}{}{}{}

\authorrunning{M.~Dantam and R.~Mayr}
\Copyright{Mohan Dantam and Richard Mayr}



\relatedversiondetails{Full version of a paper presented at MFCS 2023.}{}

\usepackage[colorinlistoftodos,prependcaption,textsize=tiny,obeyDraft]{todonotes}

\usepackage{amsmath,amssymb,amsthm}
\usepackage{mathtools}
\usepackage{thmtools} 
\usepackage{xcolor,pgf,tikz}
\usepackage{comment}
\usepackage{hyperref}
\hypersetup{colorlinks,linkcolor=blue,citecolor=blue,urlcolor=blue}
\usepackage{cleveref}
\usepackage{thm-restate}
\usepackage{array}
\usepackage{makecell}
\usepackage{stmaryrd}
\usepackage[ruled,vlined]{algorithm2e}
\usepackage{lmodern}
\usepackage{anyfontsize}
\usetikzlibrary{arrows.meta,automata,shapes,decorations,positioning,trees,calc}

\DeclareMathOperator{\counter}{\mathtt{Energy}}
\DeclareMathOperator{\runs}{\textit{Runs}}
\DeclareMathOperator{\val}{Val}
\DeclareMathOperator{\Term}{\mathtt{Term}}
\DeclareMathOperator{\FReach}{FReach}
\DeclareMathOperator{\Reach}{Reach}
\DeclareMathOperator{\Gain}{\mathtt{Gain}}
\DeclareMathOperator{\Loss}{\mathtt{Loss}}
\DeclareMathOperator{\parity}{\mathtt{PAR}}
\newcommand{\Eparity}{\mathtt{EPAR}}
\newcommand{\Oparity}{\mathtt{OPAR}}
\newcommand{\LimInf}[2]{\mathtt{LimInf}(#1 #2)}
\newcommand{\LimSup}[2]{\mathtt{LimSup}(#1 #2)}
\newcommand{\MP}[2]{\mathtt{MP}(#1\, #2)}
\newcommand{\energymove}[1]{\stackrel{#1}{\movesto}}
\newcommand{\en}{\mathtt{EN}}
\newcommand{\tendsto}{\rightarrow}

\DeclarePairedDelimiter{\set}{\{}{\}}
\DeclarePairedDelimiter{\ceil}{\lceil}{\rceil}

\newcommand{\+}[1]{\mathbb{#1}}
\newcommand{\?}[1]{\mathcal{#1}}

\newcommand{\N}{\+{N}}

\newcommand{\R}{\+{R}}
\newcommand{\x}{\times}
\newcommand{\Ocompl}{\?{O}}

\usepackage{wasysym} %
\usepackage{xparse}
\usepackage{xifthen}
\newcommand{\rsymbol}{\ocircle}
\newcommand{\zsymbol}{\Box}
\newcommand{\osymbol}{\Diamond}
\newcommand{\xsymbol}{\odot}
\newcommand{\zstates}{\states_\zsymbol}
\newcommand{\rstates}{\states_\rsymbol}
\newcommand{\ostates}{\states_\osymbol}
\newcommand{\xstates}{\states_\xsymbol}
\newcommand{\reachset}{T}

\newcommand{\win}{\mathrm{win}}
\newcommand{\lose}{\mathrm{lose}}

\newcommand{\abs}[1]{\lvert#1\rvert}

\newcommand{\size}[1]{\abs{#1}}

\newcommand{\eqby}[2][=]{\stackrel{\text{{\tiny{#2}}}}{#1}}
\newcommand{\eqdef}{\eqby{def}}

\usepackage{bm}

\newcommand{\eps}{\varepsilon}

\newcommand{\NP}{\ensuremath{\mathsf{NP}}}
\newcommand{\coNP}{\ensuremath{\mathsf{coNP}}}

\newcommand{\NEXPTIME}{\ensuremath{\mathsf{NEXPTIME}}}

\newcommand{\problemx}[3]{
\par\noindent\underline{\sc#1}\par\nobreak\vskip.2\baselineskip
\begingroup\clubpenalty10000\widowpenalty10000
\setbox0\hbox{\bf INPUT:\ }\setbox1\hbox{\bf QUESTION:\ }
\dimen0=\wd0\ifnum\wd1>\dimen0\dimen0=\wd1\fi
\vskip-\parskip\noindent
\hbox to\dimen0{\box0\hfil}\hangindent\dimen0\hangafter1\ignorespaces#2\par
\vskip-\parskip\noindent
\hbox to\dimen0{\box1\hfil}\hangindent\dimen0\hangafter1\ignorespaces#3\par
\endgroup}

\newcommand{\Prob}[2][]{\mathcal{P}^{#1}_{#2}}

\NewDocumentCommand{\valueof}{O{} O{} m}{{\mathtt{val}^{#1}_{#2}\lrc{#3}}}
\NewDocumentCommand{\limval}{O{} O{} m}{{\mathtt{Lval}^{#1}_{#2}\lrc{#3}}}
\NewDocumentCommand{\ST}{O{} O{}}{\ifthenelse{\isempty{#2}}{\ifthenelse{\isempty{#1}}{\mathtt{ST}}{\mathtt{ST}\lrc{#1}}}{\mathtt{ST}\lrc{#1,#2}}}
\NewDocumentCommand{\winset}{O{} O{}}{\IfNoValueTF{#2}{\IfNoValueTF{#1}{\mathtt{Win}}{\mathtt{Win}_{#1}}}{\IfNoValueTF{#1}{\mathtt{Win}\lrc{#2}}{\mathtt{Win}_{#1}\lrc{#2}}}}
\newcommand{\dist}{\mathcal{D}}
\newcommand{\supp}{{\sf supp}}

\newcommand{\always}{{\sf G}}
\newcommand{\eventually}{{\sf F}}
\renewcommand{\next}{{\sf X}}
\newcommand{\obj}{\mathtt{O}}

\newcommand{\hide}[1]{}

\newcommand{\lrb}[1]{\left[#1\right]}
\newcommand{\lrc}[1]{\left(#1\right)}
\newcommand{\lrd}[1]{\{#1\}}

\newcommand{\ignore}[1]{}

\newcommand{\nat}{\mathbb N}

\newcommand{\setcomp}[2]{\lrd{{#1}|\;{#2}}}
\newcommand{\tuple}[1]{\lrc{#1}}

\newcommand{\game}{{\mathcal G}}
\newcommand{\mdp}{{\mathcal M}}

\newcommand{\gametuple}{\tuple{\states,(\zstates,\ostates,\rstates),\transition,\probp}}
\newcommand{\gametupleprime}{\tuple{\states',(\zstates',\ostates',\rstates'),\transition',\probp'}}
\newcommand{\mdptuple}{\tuple{\states,\zstates,\rstates,\transition,\probp}}
\newcommand{\mdptupleprime}{\tuple{\states',\zstates',\rstates',\transition',\probp'}}

\newcommand{\states}{S}

\newcommand{\s}{s}

\newcommand{\transition}{{E}}

\newcommand{\movesto}{{\longrightarrow}}

\newcommand{\probp}{P}

\newcommand{\complementof}[1]{\overline{#1}}

\newcommand{\play}{\rho}

\newcommand{\playsof}[1]{{\it Plays}\lrc{#1}}
\newcommand{\partialplay}{\rho}

\newcommand{\zstrat}{\sigma}
\newcommand{\ostrat}{\pi}
\newcommand{\xstrat}{\tau}
\newcommand{\optzstrat}{\zstrat^*}
\newcommand{\optostrat}{\ostrat^*}

\newcommand{\zallstrats}[1]{\zstratset_{{#1}}}
\newcommand{\oallstrats}[1]{\ostratset_{{#1}}}
\newcommand{\zstratset}{\Sigma}
\newcommand{\ostratset}{\Pi}
\newcommand{\px}{\xsymbol}
\newcommand{\pz}{\zsymbol}
\newcommand{\po}{\osymbol}

\newcommand{\memory}{{\sf M}}
\newcommand{\updatefun}{upd}
\newcommand{\memconf}{{\sf m}}
\newcommand{\memconfset}{{\sf M}}
\newcommand{\memsuc}{{\sf nxt}}

\newcommand{\memup}{{\sf \updatefun}}
\newcommand{\memstrattuple}{\tuple{\memory,\initmem,\memup,\memsuc}}
\newcommand{\initmem}{\memconf_0}
\newcommand{\memstrat}[1]{{\sf T}^{#1}}

\newcommand{\om}{\omega}

\newcommand{\coloring}{{\mathit{C}ol}}

\mathchardef\mhyphen="2D %

\newcommand{\F}{{\mathcal F}}

\newcommand{\EN}[1]{\mathsf{EN}(#1)}
\newcommand{\ES}[1]{\mathsf{ST}(#1)}
\newcommand{\AS}[2][]{\mathsf{AS}^{#1}\left(#2\right)}

\newcommand{\successors}[1]{\mathsf{Succ}({#1})}

\nolinenumbers

\EventEditors{J\'{e}r\^{o}me Leroux, Sylvain Lombardy, and David Peleg}
\EventNoEds{3}
\EventLongTitle{48th International Symposium on Mathematical Foundations of Computer Science (MFCS 2023)}
\EventShortTitle{MFCS 2023}
\EventAcronym{MFCS}
\EventYear{2023}
\EventDate{August 28 to September 1, 2023}
\EventLocation{Bordeaux, France}
\EventLogo{}
\SeriesVolume{272}
\ArticleNo{36}

\begin{document}

\maketitle

\begin{abstract}
We consider simple stochastic games $\game$ with energy-parity objectives,
a combination of quantitative rewards with a qualitative parity condition.
The Maximizer tries to avoid running out of energy while simultaneously
satisfying a parity condition.
  
We present an algorithm to approximate the value of a given configuration
in 2-\NEXPTIME. Moreover, $\eps$-optimal strategies for either player
require at most ${\Ocompl}\lrc{2\textrm{-}\mathsf{EXP}\lrc{\size{\game}} \cdot\log\lrc{\frac{1}{\eps}}}$
memory modes.
\end{abstract}

\section{Introduction}\label{sec:intro}

{\bf Background.}
\emph{Simple stochastic games} (SSGs)
are 2-player turn-based perfect information stochastic games played on finite graphs.
They are also called \emph{competitive Markov decision processes} \cite{Filar_Vrieze:book},
or \emph{$2\frac{1}{2}$-player games} \cite{CJH2004,CJH2003}.
Introduced by Shapley \cite{shapley1953stochastic} in 1953, they
have since played a central role in the solution of many problems, e.g.,
synthesis of reactive systems
\cite{ramadge1987supervisory,pnueli1989synthesis}
and formal specification and verification
\cite{de2001interface,dill1989trace,alur2002alternating}.
Every state either belongs to one of the players (Maximizer or Minimizer)
or is a random state. In each round of the game the player who owns the current
state gets to choose the successor state along the game graph.
For random states the successor is chosen according to a predefined distribution.
Given a start state and strategies of Maximizer and Minimizer, this yields a
distribution over induced infinite plays.
We consider objectives $\obj$ that are measurable subsets of the set of possible plays, and the players
try to maximize (resp.\ minimize) the probability of $\obj$.

Many different objectives for SSGs have been studied in the literature.
Here we focus on parity, mean-payoff and energy objectives.
We assign numeric rewards to transitions and priorities
(aka colors), encoded by bounded non-negative numbers, to states.
A play satisfies the (min-even) {\em parity objective} iff
the minimal priority that appears infinitely often in a play is even.
It subsumes all $\omega$-regular objectives, and in particular safety,
liveness, fairness, etc.
On finite SSGs, the parity objective can be seen as a special case of
the {\em mean-payoff objective} which requires the limit average reward per
transition along a play to be positive (or\ non-negative).
Mean-payoff objectives in SSGs
go back to a 1957 paper by Gillette \cite{gillette1957stochastic} and have
been widely studied, due to their relevance for efficient control.
The {\em energy objective} \cite{chakrabarti2003resource} requires that the
accumulated reward at any time in a play stays above some finite threshold.
The intuition is that a controlled system has some finite initial energy level
that must never become depleted.
Since the accumulated reward is not bounded a-priori, this essentially turns a finite-state game into an infinite-state one.

{\bf Energy-parity.}
We consider SSGs with \emph{energy-parity} objectives,
where plays need to satisfy both an energy and a parity objective.
The parity objective specifies functional correctness, while the energy
condition can encode efficiency or risk considerations, e.g., the system
should not run out of energy since manually recharging would be costly or risky.

{\bf Previous work.}
Much work on combined objectives for stochastic systems is 
restricted to Markov decision processes (MDPs)
\cite{CD2011,chatterjee2011games,BKN2016,MSTW2017}.

For (stochastic) games, the computational complexity of single
objectives is often in $\NP\cap \coNP$, e.g., for parity or mean-payoff
objectives \cite{jurdzinski1998deciding}.
Multi-objective games can be harder, e.g.,
satisfying \emph{two different} parity objectives 
leads to $\coNP$ completeness \cite{DBLP:conf/fossacs/ChatterjeeHP07}.

Stochastic mean-payoff parity games can be solved in $\NP \cap \coNP$
\cite{CDGQ:2014}. However, this does not imply a solution for
stochastic energy-parity games, since, unlike in the non-stochastic
case \cite{CD2010}, there is no known reduction from energy-parity to mean-payoff parity
in stochastic games.
The reduction in \cite{CD2010} relies on the fact that
Maximizer has a winning \emph{finite-memory} strategy for energy-parity, which does
not generally hold for stochastic games, or even MDPs \cite{MSTW2017}.
For the same reason, the direct reduction from stochastic energy-parity to ordinary
energy games proposed in \cite{CD2011,chatterjee2011games} does not work
for general energy-parity but only for energy-B\"uchi; cf.~\cite{MSTW2017}.

Non-stochastic energy-parity games can be solved in $\NP \cap \coNP$
(and even in pseudo-quasi-polynomial time \cite{daviaud2018pseudo})
and Maximizer strategies require only finite (but exponential) memory \cite{CD2010}.

Stochastic energy-parity games have been studied in \cite{MSTW2021},
where it was shown that the almost-sure problem is decidable and in $\NP\cap \coNP$.
That is, given an initial configuration (control-state plus current energy level),
does Maximizer have a strategy to ensure
that energy-parity is satisfied with probability $1$ against any Minimizer
strategy?
Unlike in many single-objective games, such an almost-surely winning Maximizer strategy
(if it exists) requires infinite memory in general.
This holds even in MDPs and for energy-coB\"uchi objectives \cite{MSTW2017}.

However, \cite{MSTW2021} did not address quantitative questions about energy-parity
objectives, such as computing/approximating the value of a given
configuration, or the decidability of exact questions like
``Is the value of this configuration $\ge k$ ?'' for some constant $k$ (e.g., $k=1/2$).

The decidability of the latter type of exact question about the energy-parity value is open,
but there are strong indications that it is very hard.
In fact, even simpler subproblems are already at least as hard as the
\emph{positivity problem for linear recurrence sequences},
which in turn is at least as hard as the \emph{Skolem problem}
\cite{Everest:2003}. (The decidability of these problems has been open for
decades; see \cite{OW:2015} for an overview.)
Given an SSG with an energy-parity objective, suppose we remove the parity
condition (assume it is always true) and also suppose that Maximizer is passive
(does not get to make any decisions). Then we obtain an MDP where the only
active player (the Minimizer in the SSG) has a \emph{termination objective},
i.e., to reach a configuration where the energy level is $\le 0$.
Exact questions about the value of the termination objective in MDPs are
already at least as hard as the positivity problem
\cite[Section 5.2.3]{Piribauer:2021} (see also
\cite{Piribauer-Baier:2020,Piribauer-Baier:2023}).
Thus exact questions about the energy-parity value in SSGs are also
at least as hard as the positivity problem.

{\bf Our contributions.}
Since exact questions about the energy-parity value in SSGs are
positivity-hard,
we consider the problem of computing approximations of the value.
We present an algorithm that, given an SSG $\game$ and error $\eps$,
computes $\eps$-close approximations
of the energy-parity value of any given configuration in 2-\NEXPTIME.
Moreover, we show that $\eps$-optimal Maximizer (resp.\ Minimizer)
strategies can be chosen as deterministic and using only finite memory
with ${\Ocompl}\lrc{2\textrm{-}\mathsf{EXP}\lrc{\size{\game}} \cdot\log\lrc{\frac{1}{\eps}}}$
memory modes.
One can understand the idea as a constructive upper bound on the accuracy
with which the players need to remember the current energy level in
the game.
(This is in contrast to the result in \cite{MSTW2017} that
almost-surely winning Maximizer strategies require infinite
memory in general.)
Once the upper bound on Maximizer's memory for $\eps$-optimal strategies
is established, one might attempt a reduction from energy-parity to
mean-payoff parity, along similar lines as for non-stochastic games in
\cite{CD2010}. However, instead we use a more direct reduction
from energy-parity to parity in a derived SSG for our approximation algorithm.

\section{Preliminaries}\label{sec:prelim}

A \textit{probability distribution} over a countable set $S$ is a function
\mbox{$f: S \to [0,1]$} with \mbox{$\sum_{s \in S} f(s) = 1$}.
$\supp(f) \eqdef$ \mbox{$\setcomp{s}{f(s)>0}$} denotes the support of
$f$ and $\dist(S)$ is the set of all probability distributions over $S$.
Given an alphabet $\Sigma$,
let $\Sigma^{\om}$ and $\Sigma^{*}$ ($\Sigma^+$) denote the set of infinite
and finite (non-empty) sequences over $\Sigma$, respectively.
Elements of $\Sigma^{\om}$ or $\Sigma^*$ are called words.

\textbf{Games, MDPs and Markov chains.}
A \emph{Simple Stochastic Game} (SSG) is a finite-state 2-player turn-based
perfect-information stochastic game $\game = \gametuple$
where the finite set of states $\states$ is partitioned into the states $\zstates$ of
the player $\pz$ (\emph{Maximizer}),
states $\ostates$ of player $\po$ (\emph{Minimizer}),
and chance vertices (aka random states) $\rstates$.
Let $\transition \subseteq \states \x \states$ be the transition relation.
We write $\s \movesto \s'$ if $\tuple{\s,\s'} \in \transition$ and
assume that
$\successors{\s} \eqdef \{\s' \mid \s\transition{}\s'\} \neq \emptyset$
for every state $\s$.
The \emph{probability function}~$\probp$
assigns each random state $\s \in \rstates$ a distribution over
its successor states, i.e., $\probp(\s) \in \dist(\successors{\s})$.
For ease of presentation, we extend the domain of $\probp$ to
$\states^*\rstates$ by $\probp(\partialplay\s) \eqdef \probp(\s)$
for all $\partialplay\s \in \states^+\rstates$.
An \emph{MDP} is a game where one of the two players does not control any
states. An MDP is \emph{maximizing} (resp.\ \emph{minimizing})
iff $\ostates = \emptyset$ (resp.\ $\zstates = \emptyset$).
A \emph{Markov chain} is a game with only random states,
i.e., $\zstates = \ostates = \emptyset$.

\textbf{Strategies.}
A \textit{play} is an infinite sequence $\s_0\s_1 \ldots \in \states^{\omega}$
such that $\s_i \movesto \s_{i+1}$ for all $i \ge 0$.
A \textit{path} is a finite prefix of a play.
Let $\playsof{\game} \eqdef \set*{\play = \lrc{q_i}_{i \in \N} \, | q_i \movesto q_{i+1}}$
denote the set of all possible plays.
A strategy of the player $\pz$ ($\po$) is a function
$\zstrat : \states^* \zstates \to \dist(\states)$
($\ostrat : \states^* \ostates \to \dist(\states)$)
that assigns to every path
$w\s \in \states^* \zstates$ ($\in \states^* \ostates$)
a probability distribution over the successors of $\s$.
If these distributions are always Dirac then the strategy is called
\emph{deterministic} (aka pure), otherwise it is called \emph{randomized}
(aka mixed).
The set of all strategies of player $\pz$ and $\po$ in
$\game$ is denoted by $\zallstrats{\game}$ and $\oallstrats{\game}$,
respectively.
A play/path $\s_0\s_1 \ldots$ is compatible with a pair of strategies
$(\zstrat,\ostrat)$ if $\s_{i+1} \in \supp(\zstrat(\s_0 \ldots \s_i))$
whenever $\s_i \in \zstates$ and
$\s_{i+1} \in \supp(\ostrat(\s_0 \ldots \s_i))$ whenever $\s_i \in \ostates$.

Finite-memory deterministic (FD) strategies are a subclass of
strategies described by deterministic transducers
$\memstrat{} =\memstrattuple$ where $\memconfset$ is a finite set
of memory modes with initial mode $\memconf_0$,
$\memup : \memconfset \x \transition \mapsto \memconfset$
updates the memory mode upon observing a transition and
$\memsuc : \memconfset \x \xstates \mapsto \states$
chooses the successor state based on the current memory mode and state.
FD strategies without memory ($|\memconfset|=1$) are
called memoryless deterministic (MD).
For deterministic strategies, there is no difference between public memory
(observable by the other player) and private memory.

\textbf{Measure.} A game $\game$ with initial state $\s_0$ and strategies
$(\zstrat,\ostrat)$ yields a probability space
$(\s_0\states^{\om},\F_{\s_0}, \Prob[\game]{\zstrat,\ostrat,\s_0})$
where $\F_{\s_0}$ is the $\sigma$-algebra generated by the cylinder sets
$\s_0\s_1\ldots\s_n\states^{\om}$ for $n \ge 0$.
The probability measure $\Prob[\game]{\zstrat,\ostrat,\s_0}$ is
first defined on the cylinder sets.
For $\partialplay = \s_0\ldots\s_n$, let
$\Prob[\game]{\zstrat,\ostrat,\s_0}(\partialplay) \eqdef 0$ if
$\partialplay$ is not compatible with $\zstrat,\ostrat$ and
otherwise 
$\Prob[\game]{\zstrat,\ostrat,\s_0}(\partialplay\states^{\om}) \eqdef
\prod_{i=0}^{n-1} \xstrat(\s_0\ldots\s_i)(\s_{i+1}) $ where $\xstrat$
is $\zstrat$ or $\ostrat$ or $\probp$ depending on
whether $\s_i \in \zstates$ or $\ostates$ or $\rstates$, respectively.
By Carath\'eodory's extension
theorem~\cite{billingsley2008probability}, this defines a unique probability
measure on the $\sigma$-algebra.

\textbf{Objectives and Payoff functions.}
General objectives are defined by real-valued measurable functions.
However, we only consider indicator functions of measurable sets.
Hence our objectives can be described by measurable subsets $\obj \subseteq \states^{\om}$ of plays.
The payoff, under strategies $(\zstrat,\ostrat)$,
is the probability that plays belong to $\obj$.

We use the syntax and semantics of the LTL operators~\cite{CGP:book}
\mbox{$\eventually$ (eventually)}, $\always$ (always) and $\next$ (next)
to specify some conditions on plays.

\textit{Reachability \& Safety.}
A reachability objective is defined by a set of target states
$\reachset \subseteq \states$. A play $\play = \s_0s_1 \ldots$
belongs to $\eventually\,\reachset$ iff $\exists i \in \nat\, \s_i \in \reachset$.
Similarly, $\play$ belongs to $\eventually^{\le n} \reachset$
(resp.\ $\eventually^{\ge n} \reachset$) iff
$\exists i \le n$ (resp.\ $i \ge n$) such that $\s_i \in \reachset$.
Dually, the safety objective
$\always\,\reachset$ consists of all plays
which never leave $\reachset$. We have $\always\,\reachset = \neg\eventually\neg\reachset$.

\textit{Parity.} A parity objective is defined via bounded function
$\coloring: \states \to \nat$ that assigns non-negative priorities
(aka colors) to states. Given an infinite play
$\play = \s_0s_1 \ldots$, let $\text{Inf}(\play)$
denote the set of numbers that occur infinitely often in the sequence
$\coloring(\s_0)\coloring(\s_1)\ldots$.
A play $\play$ satisfies \textit{even parity}
w.r.t.\ $\coloring$ iff the minimum of $\text{Inf}(\play)$ is even.
Otherwise, $\play$ satisfies \textit{odd parity}.
The objective even parity is denoted by $\Eparity(\coloring)$
and odd parity is denoted by $\Oparity(\coloring)$.
Most of the time, we implicitly assume that the coloring function is known
and just write $\Eparity$ and $\Oparity$.
Observe that, given any coloring $\coloring$, we have
$
\complementof{\Eparity} = \Oparity$ and
$
\Oparity(\coloring) = \Eparity(\coloring + 1) 
$
where $\coloring + 1$ is the function which adds $1$ to the color of every
state. This justifies to consider only one of the even/odd parity objectives,
but, for the sake of clarity, we distinguish these objectives wherever necessary.

\textit{Energy/Reward/Counter based objectives.}
Let $r: E \to \set{-R,\dots,0,\dots,R}$ be a bounded function that assigns weights
to transitions. Depending on context, the sum of these weights in a path
can be viewed as energy, cost/reward or a counter.
If $\s \movesto \s'$ and $r((\s,\s')) = c$, we write $\s \energymove{c}
\s'$. Let $\play = \s_0 \energymove{c_0} \s_1 \energymove{c_1} \ldots$ be a play.
We say that $\play$ satisfies
\begin{enumerate}
\item
  the $k$-\textit{energy} objective $\EN{k}$ iff $\lrc{k + \sum_{i=0}^{n-1} c_i} > 0$ for all $n \ge 0$.
\item
  the \emph{$l$-storage condition} if $l+\sum_{i=m}^{n-1} c_i \ge 0$
holds for every infix $s_m \energymove{c_m} \s_{m+1}\ldots s_n$ of the play.
Let $\ES{k,l}$ denote the set of plays that satisfy both the $k$-energy
and the $l$-storage condition. Let $\ES{k} \eqdef \bigcup_l \ES{k,l}$. Clearly, $\ES{k} \subseteq \EN{k}$.
\item
  $k$-\textit{Termination} $\Term(k)$ iff there exists $n \ge 0$ such that $\lrc{k + \sum_{i=0}^{n-1} c_i} \le 0$.
\item
  \textit{Limit} objective $\LimInf{\rhd}{z}$ iff
      $\lrc{\liminf_{n \tendsto \infty}\sum_{i=0}^{n-1} c_i } \rhd\; z$
      for $\rhd \in \set{<,\le,=,\ge,>}$ and $z \in \R \cup \set{\infty,-\infty}$ and similarly for $\LimSup{\rhd}{z}$.
\item
      \textit{Mean payoff} $\MP{\rhd}{c}$ for some constant $c \in \R$
      iff $\lrc{\liminf_{n \tendsto \infty}\frac{1}{n}\sum_{i=0}^{n-1} c_i } \rhd\; c$.
\end{enumerate}
Observe that the objectives $k$-energy and $k$-termination are mutually
exclusive and cover all of the plays.
A different way to consider these objectives is to encode the energy level
(the sum of the transition weights so far) into the state space and
then consider the obtained infinite-state game with safety/reachability objective, respectively.

An objective $\obj$ is called \textit{shift-invariant}
iff for all finite paths $\partialplay$ and plays $\play' \in \states^{\omega}$,
we have $\partialplay \play' \in \obj \iff \play' \in \obj$.
Parity and mean payoff objectives are shift-invariant, but energy and
termination objectives are not.
Objective $\obj$ is called \textit{submixing} iff for all sequences of finite
non-empty words $u_0$, $v_0$, $u_1$, $v_1 \ldots$ we have
$u_0v_0u_1v_1 \ldots \in \obj \implies \left((u_0u_1\ldots \in \obj) \vee (v_0v_1 \ldots \in \obj)\right)$.

\textbf{Determinacy.}
Given an objective $\obj$ and a game $\game$, state $\s$ has value (w.r.t to $\obj$) iff 
$$\sup_{\zstrat \in \zallstrats{\game}}\inf_{\ostrat \in \oallstrats{\game}} \Prob[\game]{\zstrat,\ostrat,\s}(\obj) = \inf_{\ostrat \in \oallstrats{\game}}\sup_{\zstrat \in \zallstrats{\game}} \Prob[\game]{\zstrat,\ostrat,\s}(\obj).$$
If $\s$ has value then $\valueof[\game][\obj]{\s}$ denotes the value of
$\s$ defined by the above equality. A game with an objective is called
\textit{weakly determined} if every state has value.
Stochastic games with Borel objectives are weakly
determined~\cite{Maitra-Sudderth:2003,M1998}.
Our objectives above are Borel, hence any boolean combination of them is also
weakly determined.
For $\eps > 0$ and state $\s$, a strategy
\begin{enumerate}
    \item $\zstrat \in \zallstrats{\game}$ is $\eps$-optimal (maximizing) iff $\Prob[\game]{\zstrat,\ostrat,\s}(\obj) \ge \valueof[\game][\obj]{\s} - \eps$ for all $\ostrat \in \oallstrats{\game}$.
    \item $\ostrat \in \oallstrats{\game}$ is $\eps$-optimal (minimizing) iff $\Prob[\game]{\zstrat,\ostrat,\s}(\obj) \le \valueof[\game][\obj]{\s} + \eps$ for all $\zstrat \in \zallstrats{\game}$.
\end{enumerate}
A $0$-optimal strategy is called \emph{optimal}.
An MD strategy is called \emph{uniformly} $\eps$-optimal (resp.\ uniformly optimal)
if it is so from every start state.
An optimal strategy for player $\pz$ from state $\s$
is \emph{almost surely} winning if $\valueof[\game][\obj]{\s}=1$.
By $\AS{\obj}$ we denote the set of states that have an almost surely winning
strategy for objective $\obj$. For ease of presentation, we drop subscripts and superscripts wherever possible if they are clear from the context.
     
%

\textbf{Energy-parity.}
We are concerned with approximating the value for the combined energy-parity
objective $\en(k) \cap \Eparity$ and building $\eps$-optimal strategies.

In our constructions we use some auxiliary objectives.
Following~\cite{MSTW2021}, these are defined as
$\Gain \eqdef \LimInf{>}{-\infty}\, \cap\, \Eparity$
and
$\Loss \eqdef \complementof{\Gain} = \LimInf{=}{-\infty}\, \cup\, \Oparity$.

\begin{remark}\label{rem:ssg-md}
For finite-state SSGs and the following objectives
there exist optimal MD strategies for both players.
Moreover, if the SSG is just a maximizing MDP then the set of states
that are almost surely winning for Maximizer can be computed in polynomial
time.
\begin{enumerate}
     \item $\eventually\,\reachset$ ~\cite{CONDON1992203}
     \item
       $\LimInf{\rhd}{-\infty}$,
       $\LimInf{\rhd}{\infty}$,
       $\LimSup{\rhd}{-\infty}$,
       $\LimSup{\rhd}{\infty}$,
       $\MP{>}{0}$ ~\cite[Prop.~1]{Brazdil2010}
     \item
       $\Eparity$ ~\cite{Zielonka:1998}
     \end{enumerate}
\end{remark}

\section{The Main Result}\label{sec:result}

The following theorem states our main result.

\begin{restatable}{theorem}{thmapproxenpar}\label{thm:approx_enpar}
  Let $\game = \gametuple$ be an SSG with transition rewards in unary
  assigned by function $r$ and colors assigned to states by function $\coloring$.
  For every state $\s \in \states$, initial energy level $i \ge 0$ and error margin
  $\eps > 0$, one can compute
  \begin{enumerate}
  \item
    a rational number $v'$ such that
    $0 \le v'-\valueof[\game][\en(i)\,\cap \Eparity]{\s} \le \eps$ in
    2-\NEXPTIME.
    \footnote{We write ``computing a number $v'$ in 2-\NEXPTIME'' as a shorthand for the
      property that questions like $v' \le c$ for constants $c$
      are decidable in 2-\NEXPTIME.} 
  \item
    $\eps$-optimal FD strategies $\zstrat_\eps$ and $\ostrat_\eps$ for Maximizer
    and Minimizer, resp., in 2-\NEXPTIME.
    These strategies use
    ${\Ocompl}\lrc{2\textrm{-}\mathsf{EXP}\lrc{\size{\game}} \cdot\log\lrc{\frac{1}{\eps}}}$
    memory modes.
  \end{enumerate}
  For rewards in binary, the bounds above increase by one exponential.
\end{restatable}

\smallskip
We outline the main steps of the proof; details in the following sections.
We begin with the observation that $\en(i) \subseteq \en(j)$ for
$i \le j$, and thus for all states $\s$ we have
$
    \valueof[\game][\en(i)\,\cap\,\Eparity]{\s} \le \valueof[\game][\en(j)\,\cap\,\Eparity]{\s} \le 1.
$ 
So $\lim_{n \tendsto \infty} \valueof[\game][\en(n)\,\cap\,\Eparity]{\s}$
exists. We define
\begin{align}
    \limval[\game]{\s} &\eqdef \lim_{n \tendsto \infty} \valueof[\game][\en(n)\,\cap\,\Eparity]{\s}. \label{def:limval_energy}
\end{align}
We will see that $\limval[\game]{\s}$ and $\valueof[\game][\Gain]{\s}$ are in fact equal
(a consequence of \Cref{lem:Game-N}) and $\valueof[\game][\Gain]{\s}$
can be computed in nondeterministic polynomial time (\Cref{thm:Gain 2-player games}).
Intuitively, for high energy levels, the precise energy level does not matter
much for the value.

The main steps of the approximation algorithm are as follows.
\begin{enumerate}
\item
  Compute FD strategies $\optzstrat(\s)$
  that are optimal maximizing for the objective $\Gain$ starting from state $\s$ in $\game$.
  Compute an MD strategy $\optostrat$
  that is uniformly optimal minimizing 
  for the objective $\Gain$.
  Compute the value $\valueof[\game][\Gain]{\s}$ for every
  $\s \in \states$. 
  See \cref{sec:Gain}. \label{itm:computing gain strategies}
\item
  Compute a natural number $N$ such that for all $\s \in \states$ and
  all $i \ge N$ we have
  \[
    0 \le \valueof[\game][\Gain]{\s} - \valueof[\game][\en(i)\,\cap\,\Eparity]{\s} \le \eps.
  \]
  $N$ will be doubly exponential. See \cref{sec:Computing N}.
  \label{itm:findN}
\item
   Consider the finite-state parity game $\game'$ derived from $\game$ by encoding the energy
   level up-to $N$ into the states, i.e., the states of $\game'$ are of the form $(s,k)$
   for $s\in \states$ and $0 \le k \le N$, and colors are inherited from $\s$.
   Moreover, we add gadgets that ensure that
   states $(\s,N)$ at the upper end win with probability
   $\valueof[\game][\Gain]{\s}$ and states $(\s,0)$ at the lower end lose.
   By the previous item, $\valueof[\game][\Gain]{\s}$ is $\eps$-close to
   $\valueof[\game][\en(N)\,\cap\,\Eparity]{\s}$.
   Thus, for $k < N$ we can $\eps$-approximate the value 
   $v = \valueof[\game][\en(k)\,\cap\,\Eparity]{\s}$ by
   $v' \eqdef \valueof[\game'][\Eparity]{(\s,k)}$.
   If $k \ge N$ we can $\eps$-approximate $v$ 
   by $v' \eqdef \valueof[\game][\Gain]{\s}$.

   Moreover, we obtain $\eps$-optimal FD strategies $\zstrat_\eps$
   for Maximizer (resp.\ $\ostrat_\eps$ for Minimizer) for
   $\en(k)\,\cap\,\Eparity$ in $\game$.
   Let $\hat{\zstrat}$ (resp.\ $\hat{\ostrat}$) be optimal MD strategies
   for Maximizer (resp.\ Minimizer) for the objective $\Eparity$ in $\game'$.
   Then $\zstrat_\eps$ plays as follows.
   While the current energy level $j$ ($k$ plus the sum of the rewards so far)
   stays $<N$, then, at any state $\s'$, play like $\hat{\zstrat}$ at state
   $(\s',j)$ in $\game'$.
   Once the energy level reaches a value $\ge N$ at some state $\s'$
   for the first time, then play like $\optzstrat(s')$ forever. 
   Similarly, $\ostrat_\eps$ plays as follows.
   While the current energy level $j$ ($k$ plus the sum of the rewards so far)
   stays $<N$, then, at any state $\s'$, play like $\hat{\ostrat}$ at state
   $(\s',j)$ in $\game'$.
   Once the energy level reaches a value $\ge N$ (at any state)
   for the first time, then play like $\optostrat$ forever.
   See \cref{sec: Unfolding till N}.\label{itm:Unfolding}
\end{enumerate}

As a technical tool, we sometimes consider the dual of a game $\game$
(resp.\ the dual maximizing MDP of some minimizing MDP).
Consider $\game^{d} \eqdef \gametupleprime$
with the complement objective
$\complementof{\en(k)\,\cap\,\Eparity} = \Term(k) \cup \Oparity$,
where $\game^{d}$ is simply the game with the roles of Maximizer and
Minimizer reversed, i.e.,
\begin{align*}
    \states' &= \states &
    \zstates' &= \ostates &
    \ostates' &= \zstates &
    \rstates' &= \rstates &                       
    \transition' &= \transition &
    \probp' &= \probp
\end{align*}
Hence $\zallstrats{\game^{d}} = \oallstrats{\game}$ and
$\oallstrats{\game^{d}} = \zallstrats{\game}$.
It is easy to see that for any objective $\obj$ and start state $\s$
\begin{enumerate}
    \item $\valueof[\game][\obj]{\s} + \valueof[\game^{d}][\complementof{\obj}]{\s} = 1$.
    \item $\zstrat$ is $\eps$-optimal maximizing for $\obj$ in $\game$ iff it is $\eps$-optimal minimizing for $\complementof{\obj}$ in $\game^{d}$.
    \item $\ostrat$ is $\eps$-optimal minimizing for $\obj$ in $\game$ iff it is $\eps$-optimal maximizing for $\complementof{\obj}$ in $\game^{d}$.
\end{enumerate}
So approximating the value of $\en(k)\, \cap\, \Eparity$
in $\game$ can be reduced in linear time to approximating
the value of $\Term(k)\, \cup\, \Oparity$ in $\game^{d}$.

\section{Computing \texorpdfstring{$\valueof[\game][\Gain]{\s}$}{Value of Gain}}\label{sec:Gain}

Given an SSG $\game = \gametuple$ and a start state $\s$,
we will show how to compute $\valueof[\game][\Gain]{\s}$
and the optimal strategies for both players.

We start with the case of maximizing MDPs.
The following lemma summarizes some previous results
(\cite[Lemmas 30,16]{MSTW2021}, \cite[Lemma 26]{MSTW2017},
\cite[Proposition 4]{Gimbert2011ComputingOS}).

\begin{restatable}{lemma}{lemmaxmdp}\label{lem:maxmdp}
  Let $\mdp$ be a maximizing MDP.
  \begin{enumerate}
  \item
    $\limval[\mdp]{\s} = \valueof[\mdp][\Gain]{\s}$ for all states $\s \in \states$.
    \label{lem:maxmdp-1}
  \item
    Optimal strategies for $\Gain$ in $\mdp$ exist and can be chosen FD,
    with $\Ocompl(\exp(|\mdp|^{\Ocompl(1)}))$ memory modes,
    and exponential memory is also necessary.
    \label{lem:maxmdp-2}
  \item
    For any state $\s \in \states$, $\limval[\mdp]{\s}$ is rational and can be computed in
    $\Ocompl(|\mdp|^8)$ deterministic
    polynomial time if rewards are in unary, and in $\NP$ and $\coNP$ if
    rewards are in binary.
    \label{lem:maxmdp-3}
  \end{enumerate}
\end{restatable}
\begin{proof}
  \cref{lem:maxmdp-1} holds by \cite[Lemma 30]{MSTW2021}.

  Towards \cref{lem:maxmdp-2}, we follow the proof of \cite[Lemma 16]{MSTW2021}.
  Since $\Gain = \LimInf{>}{-\infty}\, \cap\, \Eparity$ is shift-invariant,
  there exist optimal strategies by \cite{GH2010}.
  By \cite[Theorem 18]{MSTW2017} and \cref{lem:maxmdp-1}, an optimal strategy
  for $\Gain$ can be constructed as follows.
  Let $A\eqdef\bigcup_{k\in\N}\AS{\ES{k}\cap\Eparity}$
  and $B\eqdef\AS{\LimInf{=}{\infty}\cap\Eparity}$
  be the subsets of states from which there exist almost surely winning
  strategies for the objectives $\ES{k}\cap\Eparity$ and
  $\LimInf{=}{\infty}\cap\Eparity$, respectively.
  By \cite[Theorem 8]{MSTW2017}, we can restrict the values $k$ in the
  definition of $A$ by some $k' = \Ocompl(|\states|\cdot R)$, i.e.,
  $A = \bigcup_{k \le k'}\AS{\ES{k}\cap\Eparity}$.
  An optimal strategy $\zstrat$ for $\Gain$ works in two phases.
  First it plays an optimal strategy $\zstrat_R$ towards reaching the set $A \cup B$,
  where $\zstrat_R$ can be chosen MD by \cref{rem:ssg-md}.
  Then, upon reaching $A$ (resp.\ $B$), it plays an almost surely winning
  strategy $\zstrat_A$ for the objective $\ES{k}\cap\Eparity$
  (resp.\ $\zstrat_B$ for the objective $\LimInf{=}{\infty}\cap\Eparity$).
  By \cite[Theorem 8]{MSTW2017},
  the strategy $\zstrat_A$ requires $\Ocompl(k\cdot |\states|)$ memory modes
  for a given $k$ and thus at most $\Ocompl(|\states|^2 \cdot R)$, since we
  can assume that $k \le k'$.
  Towards the strategy $\zstrat_B$, we first observe that
  in finite MDPs a strategy is almost-surely winning for $\LimInf{=}{\infty}\cap\Eparity$
  iff it is almost-surely winning for $\MP{>}{0}\cap\Eparity$.
  By \cite[Proposition 4]{Gimbert2011ComputingOS}, there exist optimal
  deterministic strategies for $\MP{>}{0}\cap\Eparity$ that use exponential
  memory, i.e., $\Ocompl(\exp(|\mdp|^{\Ocompl(1)}))$ memory modes.
  The memory required for $\zstrat_B$ exceeds that of $\zstrat_R$ and
  $\zstrat_A$ (even when $R$ is given in binary), and the one extra memory mode to
  record the switch from $\zstrat_R$ to $\zstrat_A$ (resp.\ $\zstrat_B$) is
  negligible in comparison. 
  Thus we can conclude that $\zstrat$ uses $\Ocompl(\exp(|\mdp|^{\Ocompl(1)}))$ memory modes.
  \cite[Fig.~1 and Prop.~4]{Gimbert2011ComputingOS} shows that exponential memory is necessary.
  
  Towards \cref{lem:maxmdp-3}, let $d \eqdef |\coloring(\states)|$ be the
  number of priorities in the parity condition.
  By \cite[Lemma 26]{MSTW2017}, for each $\s \in \states$,
  $\limval[\mdp]{\s}$ is rational and can be computed in deterministic time
  $\tilde{\?O}(|E| \cdot d \cdot |\states|^4 \cdot R + d \cdot |\states|^{3.5} \cdot (|P| + |r|)^2)$
  (and still in $\NP$ and $\coNP$ if $R$ is given in binary).
  So $\limval[\mdp]{\s}$ can be computed in $\Ocompl(|\mdp|^8)$ deterministic
  polynomial time if weights are given in unary, and in $\NP$ and $\coNP$ if
  weights are given in binary.
\end{proof}  

In order to extend \cref{lem:maxmdp} from MDPs to games,
we need the notion of derived MDPs, obtained by fixing the choices of one
player according to some FD strategy.
Given an SSG $\game=\gametuple$ and a finite memory deterministic (FD)
strategy $\ostrat$ for Minimizer (resp.\ $\zstrat$ for Maximizer) from a state
$\s$, described by $\memstrattuple$,
let $\game_{\ostrat}$ (resp.\ $\game^{\zstrat}$) be the maximizing
(resp.\ minimizing) MDP with state space $\memconfset \x \states$
obtained by fixing Minimizer's (resp.\ Maximizer's) choices according to
$\ostrat$ (resp.\ $\zstrat$).

\begin{lemma}\label{lem:MDP_to_game}
  For every SSG $\game$, objective $\obj$ and
  Minimizer (resp.\ Maximizer) FD strategy $\ostrat = \memstrattuple$ (resp.\ $\zstrat$),
  from state $\s$ we get
  $
    \valueof[\game^{\zstrat}][\obj]{\tuple{\initmem,\s}} \le
    \valueof[\game][\obj]{\s} \le \valueof[\game_{\ostrat}][\obj]{\tuple{\initmem,\s}}   
  $
  and equality holds if $\ostrat$ (resp.\ $\zstrat$) is optimal from state $\s$. 
\end{lemma}

\begin{theorem}\label{thm:Gain 2-player games}
Consider an SSG $\game = \gametuple$ with the $\Gain$ objective.
\begin{enumerate}
\item\label{thm:Gain 2-player games-1}
  Optimal Minimizer strategies exist and can be chosen uniform MD.
\item\label{thm:Gain 2-player games-2}
  $\valueof[\game][\Gain]{\s}$ is rational and 
  questions about it, i.e., $\valueof[\game][\Gain]{\s} \le c$ for
  constants $c$, are decidable in $\NP$.
\item\label{thm:Gain 2-player games-3}
    Optimal Maximizer strategies exist and can be chosen FD,
    with $\Ocompl(\exp(|\game|^{\Ocompl(1)}))$ memory modes.
    Moreover, exponential memory is also necessary.  
\end{enumerate}
\end{theorem}
\begin{proof}
Towards \cref{thm:Gain 2-player games-1}, observe that since both the objectives $\LimInf{=}{-\infty}$ and
$\Oparity$ are shift-invariant and submixing, so is their union, i.e.,
$\complementof{\Gain}$ is shift-invariant and submixing.
Hence, by \cite[Theorem 1.1]{GK2022},
an optimal MD strategy $\optostrat_\s$ for Minimizer
exists from any state $\s \in \states$.
Since $\states$ is finite and $\Gain$ is shift-invariant,
we can also obtain a uniformly optimal MD strategy
$\optostrat$, i.e., $\optostrat$ is optimal from every state.

Towards \cref{thm:Gain 2-player games-2},
consider the maximizing MDP $\game_{\optostrat}$ obtained from $\game$ by 
fixing $\optostrat$ (cf.~\cref{def:fix-fd}).
Since $\optostrat$ is MD, the states of $\game_{\optostrat}$ are the
same as the states as $\game$.
Since $\optostrat$ is optimal for Minimizer from every state $\s$,
we obtain that $\valueof[\game][\Gain]{\s} =
\valueof[\game_{\optostrat}][\Gain]{\s}$ for every state $\s$ by
\cref{lem:MDP_to_game}.
By \cref{lem:maxmdp}, the latter is rational and can be computed in polynomial time for
weights in unary (resp.\ in $\NP$ and $\coNP$ for weights in binary).
Thus, by guessing $\optostrat$, we can decide questions
$\valueof[\game][\Gain]{\s} \le c$ in $\NP$.

Towards \cref{thm:Gain 2-player games-3}, we again use the property
that $\complementof{\Gain}$ is shift-invariant and submixing (see above).
By~\cite[Theorem 6, Def.~24]{MSTW2021}, optimal FD Maximizer strategies for $\Gain$ in an SSG require
only $|\ostates| \cdot \lceil\log(|E|)\rceil$ many extra bits of memory above the memory required for
optimal Maximizer strategies in any derived MDP $\mdp$ where Minimizer's choices
are fixed.
Hence, by \cref{lem:maxmdp}, one can obtain optimal FD Maximizer strategies in
$\game$ that use at most
$2^{|\ostates|\cdot\lceil\log(|E|)\rceil} \cdot \Ocompl(\exp(|\mdp|^{\Ocompl(1)})) = \Ocompl(\exp(|\game|^{\Ocompl(1)}))$
memory modes.
The corresponding exponential lower bound on the memory holds already for MDPs
by \cref{lem:maxmdp}.
\end{proof}

\section{Computing the Upper Bound \texorpdfstring{$N$}{N}}\label{sec:Computing N}

We show how to compute the upper bound $N$, up-to which Maximizer needs to
remember the energy level, for any given error margin $\eps >0$.
Similarly as in \cref{sec:Gain}, we first solve the problem for maximizing
MDPs and then extend the solution to SSGs.

\subsection{Computing \texorpdfstring{$N$}{N} for maximizing MDPs}
Given a maximizing MDP $\mdp = \mdptuple$ and $\eps > 0$,
we will compute an $N \in \N$ such that for all $\s \in \states$ and all $j \ge N$
\[
	0 \le \valueof[\mdp][\Term(j)\,\cup\,\Oparity]{\s} - \valueof[\mdp][\Loss]{\s} \le \eps.
\]
Recall that $\Loss = \LimInf{=}{-\infty}\, \cup\, \Oparity$.
We now define the sets of states 
$W_0 \eqdef \AS{\Loss}$, $W_1 \eqdef \AS{\LimInf{=}{-\infty}}$ and $W_2 \eqdef \AS{\Oparity}$.
By \cref{rem:ssg-md}, there exist optimal MD strategies for
$\LimInf{=}{-\infty}$ and $\Oparity$.
Since $\Loss$ is shift-invariant and submixing, there exists an optimal MD
strategy for it by \cite[Theorem 1.1]{GK2022}.

\begin{lemma}\label{lem:W0-W1-W2}
	For every state $\s$ in the MDP $\mdp$ we have
	\begin{enumerate}
		\item $W_1 \cup W_2 \subseteq W_0$ \label{itm:Wi sub W}
		\item $\valueof[][\eventually\, {W_0}]{\s} \le \valueof[][\Loss]{\s}$\label{itm:Fw lessequal nus}
		\item $\valueof[][\Oparity \,\cap \, \complementof{\eventually\, {W_2}}]{\s} = 0$ \label{itm:val_parcapcompW2 = 0}
		\item for every initial energy level $j \ge 0$
		      \begin{align}
			       & \valueof[][\lrc{\Term(j) \, \cup \, \Oparity}\,\cap\, \eventually\, W_0]{\s} = \valueof[][\eventually\, W_0]{\s} \label{eq:valTermcapFW = valFW}                                                                                                     \\
			       & \valueof[][\Loss]{\s} \le \valueof[][\Term(j)\, \cup\, \Oparity]{\s}  \le \valueof[][\Loss]{\s} + \sup_{\zstrat}\Prob[]{\zstrat,\s} \left(\Term(j)\, \cap \, \complementof{\eventually\, {W_1}}\right) \label{eq:nus+extra}
		      \end{align}
	\end{enumerate}
\end{lemma}
\begin{proof}\,
\begin{enumerate}
\item
This follows directly from the definitions of $W_0,W_1,W_2$.
\item
  Let $\zstrat'$ be an optimal MD strategy for $\eventually\, {W_0}$ from $\s$ and
  $\zstrat''$ be an almost surely winning MD strategy for $\Loss$ from any state in $W_0$.
  Let $\zstrat$ be the strategy that plays $\zstrat'$ until reaching $W_0$ and
  then switches to $\zstrat''$.
  We have $\valueof[][\Loss]{\s} \ge \Prob{\zstrat,\s}(\Loss) \ge
  \Prob{\zstrat',\s}(\eventually\,W_0)
  = \valueof[][\eventually\, W_0]{\s}$.
\item
  For $\s \in W_2$ the statement is obvious.
  So let $\s \notin W_2$ and consider the modified
  MDP $\mdp' = \mdptupleprime$ where all states in $W_2$ are collapsed into a
  losing sink.
  I.e., $\states' \eqdef (\states \setminus W_2) \uplus \set*{\mathit{trap}}$,
  with $\mathit{trap}$ a new random sink state having color $0$
  (thus losing for objective $\Oparity$),
  $\transition'$ contains all of
  $\left(\transition \cap \set*{(\states \setminus W_2)\x(\states \setminus
      W_2)} \cup (\mathit{trap},\mathit{trap}) \right)$ and all transitions to
  $W_2$ are redirected to $\mathit{trap}$ and $\probp'$ is derived accordingly
  from $\probp$.
  Then $\valueof[\mdp'][\Oparity]{\hat{\s}} = \valueof[\mdp][\Oparity \,\cap \, \complementof{\eventually {W_2}}]{\hat{\s}}$
  for all states $\hat{\s} \in \states \setminus W_2$. 
  Towards a contradiction, assume that
  $\valueof[\mdp][\Oparity \,\cap \, \complementof{\eventually {W_2}}]{\s} > 0$.
  Hence $\valueof[\mdp'][\Oparity]{\s} > 0$.
  Then, by \cite[ Theorem~3.2]{GH2010}, there exists a state $\s' \in \states'$
  such that $\valueof[\mdp'][\Oparity]{\s'} = 1$,
  and it is easy to see that $\s' \neq \mathit{trap}$ and thus
  $\s' \in \states \setminus W_2$.
  But this implies that $\valueof[\mdp][\Oparity]{\s'} = 1$ and thus $\s' \in W_2$, a contradiction.
\item
  Let $\obj \eqdef \Term(j)\,\cup\, \Oparity$.
  For \cref{eq:valTermcapFW = valFW}, the first inequality
  $\valueof[][\obj\,\cap\, \eventually W_0]{\s} \le \valueof[][\eventually
  W_0]{\s}$ is trivial, since $\obj\, \cap\, \eventually\,W_0 \subseteq \eventually W_0$.
  To show the reverse inequality, consider the strategy $\zstrat$
  that first plays like an optimal MD strategy $\zstrat'$ for the objective
  $\eventually\,W_0$ and after reaching $W_0$ switches to an almost surely winning MD
  strategy $\zstrat''$ for the objective $\Loss$.
  Then
  $
  \valueof[][\obj\,\cap\,\eventually W_0]{\s}
  \ge
  \Prob[]{\zstrat,\s}(\obj\,\cap\,\eventually \,W_0)
  \ge
  \Prob[]{\zstrat,\s}(\Loss\,\cap\,\eventually \,W_0)
  =
  \Prob[]{\zstrat',\s}(\eventually \,W_0)
  =
  \valueof[][\eventually W_0]{\s}
  $,
  where the second inequality is due to $\LimInf{=}{-\infty}\subseteq \Term(j)$.

  For~\cref{eq:nus+extra}, the first inequality is again due to the fact
  that $\LimInf{=}{-\infty}\subseteq \Term(j)$ for all $j \ge 0$.
  Towards the second inequality of \cref{eq:nus+extra} we have
  \begin{align*}
    & \valueof[][\obj]{\s} \\
    & = \sup_{\zstrat}\Prob[]{\zstrat,\s}(\obj)\\ 
    & = \sup_{\zstrat}\left(\Prob{\zstrat,\s}\left(\obj\, \cap\, \eventually\, {W_0}\right) + \Prob{\zstrat,\s}\left(\obj\, \cap\,
      \complementof{\eventually\, {W_0}}\right)  \right)        & \lrc{\textrm{Law of total probability}} \nonumber \\
    & \le \sup_{\zstrat}\Prob{\zstrat,\s}\left(\obj\, \cap\, \eventually\, {W_0}\right) + \sup_{\zstrat}\Prob{\zstrat,\s}\left(\obj\, \cap\, \complementof{\eventually\, {W_0}}\right) & \lrc{\sup\lrc{f+g} \le \sup f + \sup g} \nonumber \\
    & = \sup_{\zstrat}\Prob{\zstrat,\s}\left(\eventually\, {W_0}\right)+\sup_{\zstrat}\Prob{\zstrat,\s} \left(\obj\, \cap\,
      \complementof{\eventually\, {W_0}}\right)                                  & \lrc{\cref{eq:valTermcapFW = valFW}} \nonumber          \\
    & \le \valueof[][\Loss]{\s}  + \sup_{\zstrat}\Prob{\zstrat,\s}
      \left(\obj\, \cap\, \complementof{\eventually\, {W_0}}\right)
    & \lrc{\cref{itm:Fw lessequal nus}}
  \end{align*}
  We can upper-bound the second summand above as follows.
  \begin{align*}
    & \sup_{\zstrat}\Prob{\zstrat,\s} (\obj\, \cap\,  \complementof{\eventually\, {W_0}}) \\
    & = \sup_{\zstrat}\Prob{\zstrat,\s} \left((\Term(j) \, \cup \, \Oparity)\, \cap\, \complementof{\eventually\, {W_0}}\right)\\
    & \le \sup_{\zstrat}\Prob{\zstrat,\s} \left(\Term(j)\, \cap\,
      \complementof{\eventually\, {W_0}}\right) +
      \sup_{\zstrat}\Prob{\zstrat,\s} \left(\Oparity\, \cap\,
      \complementof{\eventually\, {W_0}}\right)    & \lrc{\textrm{Union bound}} \nonumber   \\
    & \le \sup_{\zstrat}\Prob{\zstrat,\s} \left(\Term(j)\, \cap\, \complementof{\eventually\, {W_1}}\right) + \sup_{\zstrat}\Prob{\zstrat,\s} \left(\Oparity\, \cap\, \complementof{\eventually\, {W_2}}\right) & \lrc{\cref{itm:Wi sub W}}\quad\quad \nonumber                         \\
    & = \sup_{\zstrat}\Prob{\zstrat,\s} \left(\Term(j)\, \cap\, \complementof{\eventually\, {W_1}}\right) & \lrc{\cref{itm:val_parcapcompW2 = 0}}\quad\quad\qedhere
  \end{align*}
\end{enumerate}
\end{proof}

We show that the term
$\sup_{\zstrat}\Prob[]{\zstrat,\s} \left(\Term(j)\, \cap \, \complementof{\eventually\, {W_1}}\right)$
in \cref{eq:nus+extra} can be made arbitrarily small for large $j$.
To this end, we use \cite[Lemma~3.9]{BBEK:IC2013} (adapted to our notation).

\begin{lemma}~\cite[Lemma~3.9 and Claim~6]{BBEK:IC2013} \label{lem:exp_bound_term_value}
Let $\mdp = \mdptuple$ be a maximizing finite MDP with rewards in unary and
$W_1 \eqdef \AS{\LimInf{=}{-\infty}}$.
One can compute, in polynomial time, a rational constant $c < 1$,
and an integer $h \ge 0$
such that for all $j \ge h$
and $\s \in \states$
	\[
		\sup_{\zstrat}\Prob{\zstrat,\s}\lrc{\Term(j) \, \cap\,
                  \complementof{\eventually\,W_1}} \le
                \frac{c^j}{1-c}.
	\]
Moreover, $1/(1-c) \in {\Ocompl}\left(\exp(\size{\mdp}^{\Ocompl(1)})\right)$ and
$h \in {\Ocompl}(\exp\lrc{\size{\mdp}^{{\Ocompl}(1)}})$.
\end{lemma}

\bigskip
\begin{restatable}{lemma}{lemMDPN}\label{lem:MDP-N}
Consider a maximizing MDP $\mdp = \mdptuple$, $\eps >0$ and the constants $c,h$ from \cref{lem:exp_bound_term_value}.
For rewards in unary and $i \ge N$ we have
$\valueof[\mdp][\Term(i)\,\cup\, \Oparity]{\s} - \valueof[\mdp][\Loss]{\s} \le
\eps$ where
$N \eqdef \max\lrc{h,\ceil{\log_c\lrc{\eps \cdot \lrc{1-c}}}} \in {\Ocompl}\left(\exp(\size{\mdp}^{\Ocompl(1)})\cdot \log\lrc{1/\eps}\right)$.

For rewards in binary we have
$N \in {\Ocompl}\left(\exp(\exp(\size{\mdp}^{\Ocompl(1)}))\cdot
  \log\lrc{1/\eps}\right)$, i.e.,
the size of $N$ increases by one exponential.
\end{restatable}
\begin{proof}[Proof sketch.] 
  For rewards in unary, the result follows from
  \cref{lem:W0-W1-W2}(\cref{eq:nus+extra}) and \cref{lem:exp_bound_term_value}.
  For rewards in binary, the constants increase by one exponential via
  encoding binary rewards into unary rewards in a modified MDP.
\end{proof}

\subsection{Computing \texorpdfstring{$N$}{N} for SSGs}\label{subsec:N-for-SSG}
In order to compute the bound $N$ for an SSG $\game$,
we first consider bounds $N(\s)$ for individual states $\s$
and then take their maximum.
Given a state $\s$,
we can use \cref{thm:Gain 2-player games}(\cref{thm:Gain 2-player games-3})
to obtain an optimal FD strategy (with $\Ocompl(\exp(|\game|^{\Ocompl(1)}))$ memory modes)
$\optzstrat(\s) = \memstrattuple$ for Maximizer from state $\s$ w.r.t.\ the
$\Gain$ objective.
\cref{thm:Gain 2-player games}(\cref{thm:Gain 2-player games-1})
yields a uniform MD strategy $\optostrat$
that is optimal for Minimizer from all states $\s$ w.r.t.\ the
$\Gain$ objective.

\begin{restatable}{lemma}{lemGameN}\label{lem:Game-N}
Given an SSG $\game=\gametuple$ and $\eps >0$, we can compute a number $N \in \N$
such that for all $i \ge N$ and states $\s \in \states$ we have
\begin{equation}\label{lem:Game-N-a}
  \valueof[\game][\en(i)\, \cap\, \Eparity]{\s} - \eps
  \le
  \valueof[\game][\Gain]{\s} - \eps
  \le
  \inf_\ostrat \Prob[\game]{\optzstrat(\s),\ostrat,\s}\lrc{\en(i)\,\cap\,\Eparity}
  \le
  \valueof[\game][\en(i)\,\cap\,\Eparity]{\s}
\end{equation}
i.e., $\optzstrat(\s)$ is $\eps$-optimal for Maximizer for $\en(i)\,\cap\,\Eparity$ for all $i \ge N$.
In particular, $0 \le \valueof[\game][\Gain]{\s} - \valueof[\game][\en(i)\,\cap\,\Eparity]{\s} \le \eps$.

Moreover, $\optostrat$ is $\eps$-optimal for Minimizer from any state $\s$ for $i \ge N$.
\begin{equation}\label{lem:Game-N-b}
  \sup_{\zstrat}\Prob[\game]{\zstrat,\optostrat,\s}\lrc{\en(i)\,\cap\,\Eparity}
  \le
  \sup_{\zstrat}\Prob[\game]{\zstrat,\optostrat,\s}\lrc{\Gain}
  =
  \valueof[\game][\Gain]{\s}
  \le
  \valueof[\game][\en(i)\,\cap\,\Eparity]{\s} + \eps
\end{equation}
For rewards in unary, $N$ is doubly exponential, i.e.,
$N \in {\Ocompl}\left(\exp(\exp(\size{\game}^{\Ocompl(1)}))\cdot \log\lrc{1/\eps}\right)$
and it can be computed in exponential time.
For rewards in binary, the size of $N$ and its computation time increase
by one exponential, respectively.
\end{restatable}
\begin{proof}
Assume that rewards are in unary.
The first inequality of \eqref{lem:Game-N-a} holds because 
$\en(i)\, \cap\, \Eparity \subseteq \Gain$ for any $i$.
The third inequality of \eqref{lem:Game-N-a} follows from the definition of
the value.
Towards the second inequality of \eqref{lem:Game-N-a},
we consider the minimizing MDP
$\mdp(\s) \eqdef \game^{\optzstrat(\s)}$ 
obtained by fixing the Maximizer strategy $\optzstrat(\s)$.
Since $\optzstrat(\s)$ is optimal for Maximizer from state $\s$
wrt.\ the objective $\Gain$, \Cref{lem:MDP_to_game} yields that
\begin{equation}\label{eq:lem:Game-N-1}
\valueof[\game][\Gain]{\s} = \valueof[\mdp(\s)][\Gain]{\tuple{\initmem,\s}}.
\end{equation}
Since $\optzstrat(\s)$ has $\Ocompl(\exp(|\game|^{\Ocompl(1)}))$ memory modes,
the size of $\mdp(\s)$ is exponential in $|\game|$ and $\mdp(\s)$ can be computed in
exponential time.

Now we consider the dual maximizing MDP $\mdp(\s)^d$ and the objectives
$\Term(i)\, \cup\, \Oparity$ and $\Loss$. (Note that $\mdp(\s)^d$ has the same
size as $\mdp(\s)$.)
From \Cref{lem:MDP-N}, we obtain
a bound $N(\s) \in \N$ such that for all $i \ge N(\s)$
\begin{equation}\label{eq:lem:Game-N-2}
0\le \valueof[\mdp(\s)^d][\Term(i)\, \cup\, \Oparity]{\tuple{\initmem,\s}} -\valueof[\mdp(\s)^d][\Loss]{\tuple{\initmem,\s}} \le \eps.
\end{equation}
By \Cref{lem:MDP-N} and \Cref{lem:exp_bound_term_value},
$N(\s)$ is exponential in $|\mdp(\s)^d|$ and thus doubly exponential
in $|\game|$, i.e.,
$N(\s) \in {\Ocompl}\left(\exp(\exp(\size{\game}^{\Ocompl(1)}))\cdot \log\lrc{1/\eps}\right)$.
Moreover, $N(\s)$ can be computed in time polynomial in $|\mdp(\s)^d|$
and thus in time exponential in $|\game|$.
By duality, we can rewrite \Cref{eq:lem:Game-N-2} for $\mdp(\s)$ as follows.
For all $i \ge N(\s)$
\begin{equation}\label{eq:lem:Game-N-3}
  0
  \le
  \valueof[\mdp(\s)][\Gain]{\tuple{\initmem,\s}} -
  \valueof[\mdp(\s)][\en(i)\, \cap\, \Eparity]{\tuple{\initmem,\s}}
  \le
  \eps.
\end{equation}
  In order to get a uniform upper bound that holds for all states,
  let $N \eqdef \max_{\s \in \states} N(\s)$.
  Since $|\states|$ is linear, we still have 
  $N \in {\Ocompl}\left(\exp(\exp(\size{\game}^{\Ocompl(1)}))\cdot \log\lrc{1/\eps}\right)$
  and it can be computed in exponential time in $|\game|$.
  Finally, we can show the second inequality of \eqref{lem:Game-N-a}.
  \begin{align*}
    & \inf_\ostrat \Prob[\game]{\optzstrat(\s),\ostrat,\s}\lrc{\en(i)\,\cap\,\Eparity}\\
    & = \inf_\ostrat \Prob[\mdp(\s)]{\ostrat,\tuple{\initmem,\s}}\lrc{\en(i)\,\cap\,\Eparity}\\
    & = \valueof[\mdp(\s)][\en(i)\,\cap\,\Eparity]{\tuple{\initmem,\s}}\\
    & \ge \valueof[\mdp(\s)][\Gain]{\tuple{\initmem,\s}} - \eps
    & \mbox{by $i \ge N \ge N(\s)$ and \Cref{eq:lem:Game-N-3}}\\
    & = \valueof[\game][\Gain]{\s} - \eps
    & \mbox{by \eqref{eq:lem:Game-N-1}}
  \end{align*}

 The first inequality of \eqref{lem:Game-N-b} holds because 
 $\en(i)\, \cap\, \Eparity \subseteq \Gain$ for any $i$.
 The equality in \eqref{lem:Game-N-b} holds by the
 optimality of $\optostrat$.
 The second inequality of \eqref{lem:Game-N-b} follows from the previously
 stated consequence of \eqref{lem:Game-N-a}.
   
For rewards in binary, the sizes of the numbers $N(\s)$ (and hence $N$)
and the time to compute it increase by one exponential by \Cref{lem:MDP-N}.
\end{proof}

\section{Unfolding the Game to Energy Level  \texorpdfstring{$N$}{N}}\label{sec: Unfolding till N}

Given an SSG $\game = \gametuple$ and error tolerance $\eps >0$,
for each state $\s \in \states$ and energy level $i \ge 0$,
we want to compute a rational number
$v'$ which satisfies
$0 \le v'- \valueof[\game][\en(i)\, \cap \, \Eparity]{\s} \le \eps$,
and $\eps$-optimal FD strategies $\zstrat_\eps$ and $\ostrat_\eps$ for
Maximizer and Minimizer, resp.
We achieve this by constructing a finite-state parity game $\game'$
that closely approximates the original game $\game$,
as described in \Cref{sec:result}(\cref{itm:Unfolding}).

For clarity, we explain the construction in two steps.
In the first step, we consider a finite-state parity game $\game\lrb{N}$.
(Unlike $\game'$, the game $\game\lrb{N}$ is not actually constructed.
It just serves as a part of the correctness proof.)
$\game\lrb{N}$ encodes the energy level up-to $N+R$ (where $R$ is the maximal
transition reward) into the states, i.e., it has states
of the form $(\s,k)$ with $k \le N+R$.
It imitates the original game $\game$ till energy level $N+R$,
but at any state $\lrc{\s,i}$ with energy level $i \ge N$ it jumps
to a winning state with probability $\valueof[\game][\en(i)\, \cap \,
\Eparity]{\s}$ and to a losing state with probability
$1-\valueof[\game][\en(i)\, \cap \, \Eparity]{\s}$.
(We need the margin up-to $N+R$, because transitions can have rewards $>1$, so
the level $N$ might not be hit exactly.)
Similarly, at states $\lrc{\s,0}$ with energy level $0$,
we jump to a losing state.
The coloring function in the new game $\game\lrb{N}$
derives its colors from the colors in the original game
$\game$, i.e., all states $\lrc{\s,i}$ have the same color as $\s$ in $\game$.

By construction of $\game\lrb{N}$, for $i \le N$, the $\Eparity$ value of $\lrc{\s,i}$ in
$\game\lrb{N}$ coincides with $\valueof[\game][\en(i)\, \cap\, \Eparity]{\s}$.

In the second step, since we do not know the exact values
$\valueof[\game][\en(i)\, \cap \, \Eparity]{\s}$ for $N+R \ge i > N$,
we approximate these by the slightly larger $\valueof[\game][\Gain]{\s}$.
I.e., we modify $\game\lrb{N}$ by replacing the probability values
$\valueof[\game][\en(i)\, \cap \, \Eparity]{\s}$ for the jumps to the winning
state by $\valueof[\game][\Gain]{\s}$. Let $\game'$ be the resulting
finite-state parity game.
It follows from \Cref{lem:Game-N} that
$0 \le \valueof[\game][\Gain]{\s} -
\valueof[\game][\en(i)\,\cap\,\Eparity]{\s} \le \eps$ for $i \ge N$ and
$\limval[\game][\en\, \cap \, \Eparity]{\s} = \valueof[\game][\Gain]{\s}$.
Thus $\game'$ $\eps$-over-approximates $\game\lrb{N}$ and $\game$, and we obtain the
following lemma.

\begin{lemma}\label{lem:approxenparvalueslessthanN}
For all states $\s$ and all $0 \leq i \leq N$
\begin{align*}
        &\valueof[\game\lrb{N}][\Eparity]{\lrc{\s,i}} =
          \valueof[\game][\en(i)\, \cap \, \Eparity]{\s}, \mbox{and}\\
        & 0 \le \valueof[\game^\prime][\Eparity]{\lrc{\s,i}} - \valueof[\game\lrb{N}][\Eparity]{\lrc{\s,i}} \leq \eps.
    \end{align*}
\end{lemma}

Now we are ready to prove the main theorem.

\thmapproxenpar*
\begin{proof}
For $i > N$ we output $v' = \valueof[\game][\Gain]{\s}$, which satisfies the
condition by \Cref{lem:Game-N}.
For $i \le N$ we output $v' = \valueof[\game'][\Eparity]{\lrc{\s,i}}$, which
satisfies the condition by \Cref{lem:approxenparvalueslessthanN}.
By \Cref{thm:Gain 2-player games}, the values $\valueof[\game][\Gain]{\s}$ are
rational for all states $\s$. Therefore all probability values in $\game'$ are
rational and thus the $\Eparity$ values of all states in $\game'$ are rational. 
Hence our numbers $v'$ are always rational.

By \Cref{thm:Gain 2-player games}, the values $\valueof[\game][\Gain]{\s}$ for
all states $\s \in \states$ can be computed in exponential time.
By \Cref{lem:Game-N}, $N \in
{\Ocompl}\left(\exp(\exp(\size{\game}^{\Ocompl(1)}))\cdot \log\lrc{1/\eps}\right)$
is doubly exponential.
Therefore, we can construct $\game'$ in 
${\Ocompl}\left(\exp(\exp(\size{\game}^{\Ocompl(1)}))\cdot
  \log\lrc{1/\eps}\right)$ time and space.
Questions about the parity values of states in $\game'$ can be decided in
nondeterministic time polynomial in $|\game'|$. Thus the numbers $v'$ are computed in 2-\NEXPTIME.

Towards Item 2, we construct $\eps$-optimal FD strategies $\zstrat_\eps$
for Maximizer (resp.\ $\ostrat_\eps$ for Minimizer) for
$\en(i)\,\cap\,\Eparity$ in $\game$.
Let $\hat{\zstrat}$
(resp.\ $\hat{\ostrat}$)
be optimal MD strategies
for Maximizer (resp.\ Minimizer) for the objective $\Eparity$ in $\game'$,
which exist by \Cref{rem:ssg-md}.
Since these strategies are MD, they can be guessed in nondeterministic time
polynomial in the size $|\game'|$, and thus 
in ${\Ocompl}\left(\exp(\exp(\size{\game}^{\Ocompl(1)}))\cdot
  \log\lrc{1/\eps}\right)$ nondeterministic time.

Then $\zstrat_\eps$ plays as follows.
While the current energy level $j$ ($i$ plus the sum of the rewards so far)
stays $<N$, then, at any state $\s'$, play like $\hat{\zstrat}$ at state
$(\s',j)$ in $\game'$.
Once the energy level reaches a value $\ge N$ at some state $\s'$
for the first time, then play like $\optzstrat(s')$ forever.
(Recall that $\optzstrat(s')$ is the optimal FD Maximizer strategy for $\Gain$
from state $s'$ from \cref{subsec:N-for-SSG}.)
$\zstrat_\eps$ is $\eps$-optimal by \Cref{lem:approxenparvalueslessthanN}
and \Cref{lem:Game-N}.
It needs to remember the energy level up-to $N$ while simulating
$\hat{\zstrat}$. Moreover, $\optzstrat(s')$ needs
$\Ocompl(\exp(|\game|^{\Ocompl(1)}))$ memory modes by \Cref{thm:Gain 2-player games}.
Finally, it needs to remember the switch from $\hat{\zstrat}$ to $\optzstrat(s')$.
Since $N \in
{\Ocompl}\left(\exp(\exp(\size{\game}^{\Ocompl(1)}))\cdot \log\lrc{1/\eps}\right)$
dominates the rest, $\zstrat_\eps$ uses
${\Ocompl}\left(\exp(\exp(\size{\game}^{\Ocompl(1)}))\cdot \log\lrc{1/\eps}\right)$ memory modes.

Similarly, $\ostrat_\eps$ plays as follows.
While the current energy level $j$ 
stays $<N$, at any state $\s'$, play like $\hat{\ostrat}$ at state
$(\s',j)$ in $\game'$.
Once the energy level reaches a value $\ge N$ (at any state)
for the first time, then play like $\optostrat$ forever
(where $\optostrat$ is the uniform optimal MD Minimizer strategy for 
$\Gain$ from \cref{subsec:N-for-SSG}.)
$\ostrat_\eps$ is $\eps$-optimal by \Cref{lem:approxenparvalueslessthanN}
and \Cref{lem:Game-N}.
While $\optostrat$ is MD and does not use any memory, $\ostrat_\eps$ still
needs to remember the energy level up-to $N$
while simulating $\hat{\ostrat}$,
and thus it uses
${\Ocompl}\left(\exp(\exp(\size{\game}^{\Ocompl(1)}))\cdot \log\lrc{1/\eps}\right)$ memory modes.

For rewards in binary, all bounds increase by one exponential via an encoding
of $\game$ into an exponentially larger but equivalent game with rewards in unary.
\end{proof}

No nontrivial lower bounds are known on the computational complexity of
approximating $\valueof[\game][\en(i)\,\cap \Eparity]{\s}$.
However, even without the parity part, the problem appears to be hard.
The best known algorithm for approximating the value of the energy objective
(resp.\ the dual termination objective) runs in $\NEXPTIME$ for SSGs with
rewards in unary \cite{BBEK:IC2013}.

As for lower bounds on the strategy complexity, 
$\eps$-optimal Maximizer strategies need at least an exponential number of
memory modes (for any $0 < \eps < 1$) even in maximizing MDPs.
This can easily be shown by extending the example in
\Cref{lem:maxmdp}(\Cref{lem:maxmdp-2}) and
\cite[Fig.~1 and Prop.~4]{Gimbert2011ComputingOS}
that shows the lower bound for the $\Gain$ objective.
First loop in a state with an unfavorable color to accumulate a
sufficiently large reward (depending on $\eps$) and then switch to the MDP in 
\cite[Fig.~1 and Prop.~4]{Gimbert2011ComputingOS} to play for $\Gain$
(since $\en(i)\,\cap\, \Eparity$ will be very close to $\Gain$ then).
Even the latter part requires exponentially many memory modes.

\section{Conclusion \& Extensions}\label{sec:conclusion}

We gave a procedure to compute $\eps$-approximations of the value of combined
energy-parity objectives in SSGs.
The decidability of questions about the exact values is open, but the problem
is at least as hard as the positivity problem for linear recurrence sequences
\cite[Section 5.2.3]{Piribauer:2021}.
%
Unlike almost surely winning Maximizer strategies which
require infinite memory in general \cite{MSTW2017,MSTW2021},
$\eps$-optimal strategies for either player require only finite memory
with at most doubly exponentially many memory modes.

An interesting topic for further study is whether these results can be extended
to other combined objectives where the parity part is replaced by something
else, i.e., energy-X for some objective X
(e.g., some other color-based condition like Rabin/Streett, or a quantitative
objective about multi-dimensional transition rewards).
While our proofs are not completely
specific to parity, they do use many strong properties that parity satisfies. 

\begin{itemize}
\item
  Shift-invariance of $\Eparity$ is used in several places,
  e.g.\ in \Cref{lem:W0-W1-W2} (and thus its consequences) and
  for the correctness of the constructions in \Cref{sec: Unfolding till N}.
\item
  We use the fact that $\Eparity$ goes well together with
  $\LimInf{>}{-\infty}$,
  i.e., the objective $\Gain = \LimInf{>}{-\infty}\, \cap\, \Eparity$ allows
  optimal FD strategies for Maximizer in MDPs; cf.~\Cref{lem:maxmdp}.
\item
  The submixing property of $\Oparity = \complementof{\Eparity}$ is used in \cref{thm:Gain 2-player games}
  to lift \Cref{lem:maxmdp} from MDPs to SSGs.
\end{itemize}

\newpage
\bibliographystyle{plainurl}
\bibliography{conferences,journals,ref1}

\begin{thebibliography}{10}

\bibitem{alur2002alternating}
Rajeev Alur, Thomas~A Henzinger, and Orna Kupferman.
\newblock Alternating-time temporal logic.
\newblock {\em Journal of the ACM}, 49(5):672--713, 2002.

\bibitem{billingsley2008probability}
Patrick Billingsley.
\newblock {\em Probability and measure}.
\newblock John Wiley \& Sons, 2008.

\bibitem{BBEK:IC2013}
T.~Br\'{a}zdil, V.~Bro\v{z}ek, K.~Etessami, and A.~Ku\v{c}era.
\newblock {Approximating the Termination Value of One-Counter MDPs and
  Stochastic Games}.
\newblock {\em Information and Computation}, 222:121--138, 2013.
\newblock URL: \url{https://arxiv.org/abs/1104.4978}.

\bibitem{BKN2016}
T.~Br\'{a}zdil, A.~Ku\v{c}era, and P.~Novotn\'{y}.
\newblock Optimizing the expected mean payoff in energy {Markov} decision
  processes.
\newblock In {\em International Symposium on Automated Technology for
  Verification and Analysis (ATVA)}, volume 9938 of {\em LNCS}, pages 32--49,
  2016.

\bibitem{Brazdil2010}
Tom{\'a}s Br{\'a}zdil, V{\'a}clav Brozek, and Kousha Etessami.
\newblock {One-Counter Stochastic Games}.
\newblock In Kamal Lodaya and Meena Mahajan, editors, {\em IARCS Annual
  Conference on Foundations of Software Technology and Theoretical Computer
  Science (FSTTCS 2010)}, volume~8 of {\em Leibniz International Proceedings in
  Informatics (LIPIcs)}, pages 108--119, Dagstuhl, Germany, 2010. Schloss
  Dagstuhl--Leibniz-Zentrum fuer Informatik.
\newblock Full version at \url{http://arxiv.org/abs/1009.5636}.
\newblock URL: \url{http://drops.dagstuhl.de/opus/volltexte/2010/2857}, \href
  {http://dx.doi.org/10.4230/LIPIcs.FSTTCS.2010.108}
  {\path{doi:10.4230/LIPIcs.FSTTCS.2010.108}}.

\bibitem{chakrabarti2003resource}
Arindam Chakrabarti, Luca De~Alfaro, Thomas~A Henzinger, and Mari{\"e}lle
  Stoelinga.
\newblock Resource interfaces.
\newblock In {\em International Workshop on Embedded Software}, pages 117--133,
  2003.

\bibitem{CD2010}
Krishnendu Chatterjee and Laurent Doyen.
\newblock Energy parity games.
\newblock In {\em International Colloquium on Automata, Languages and
  Programming (ICALP)}, volume 6199 of {\em LNCS}, pages 599--610, 2010.

\bibitem{CD2011}
Krishnendu Chatterjee and Laurent Doyen.
\newblock Energy and mean-payoff parity {Markov} decision processes.
\newblock In {\em International Symposium on Mathematical Foundations of
  Computer Science (MFCS)}, volume 6907, pages 206--218, 2011.

\bibitem{chatterjee2011games}
Krishnendu Chatterjee and Laurent Doyen.
\newblock Games and {M}arkov decision processes with mean-payoff parity and
  energy parity objectives.
\newblock In {\em Mathematical and Engineering Methods in Computer Science
  (MEMICS)}, volume 7119 of {\em LNCS}, pages 37--46. Springer, 2011.

\bibitem{CDGQ:2014}
Krishnendu Chatterjee, Laurent Doyen, Hugo Gimbert, and Youssouf Oualhadj.
\newblock Perfect-information stochastic mean-payoff parity games.
\newblock In {\em International Conference on Foundations of Software Science
  and Computational Structures (FoSSaCS)}, volume 8412 of {\em LNCS}, 2014.

\bibitem{DBLP:conf/fossacs/ChatterjeeHP07}
Krishnendu Chatterjee, Thomas~A. Henzinger, and Nir Piterman.
\newblock Generalized parity games.
\newblock In Helmut Seidl, editor, {\em International Conference on Foundations
  of Software Science and Computational Structures (FoSSaCS)}, volume 4423 of
  {\em LNCS}, pages 153--167. Springer, 2007.
\newblock URL: \url{https://doi.org/10.1007/978-3-540-71389-0\_12}, \href
  {http://dx.doi.org/10.1007/978-3-540-71389-0\_12}
  {\path{doi:10.1007/978-3-540-71389-0\_12}}.

\bibitem{CJH2003}
Krishnendu Chatterjee, Marcin Jurdzi{\'{n}}ski, and Thomas~A. Henzinger.
\newblock Simple stochastic parity games.
\newblock In {\em Computer Science Logic (CSL)}, volume 2803 of {\em LNCS},
  pages 100--113. Springer, 2003.

\bibitem{CJH2004}
Krishnendu Chatterjee, Marcin Jurdzi\'{n}ski, and Thomas~A. Henzinger.
\newblock Quantitative stochastic parity games.
\newblock In {\em ACM-SIAM Symposium on Discrete Algorithms (SODA)}, pages
  121--130. SIAM, 2004.

\bibitem{CGP:book}
E.M. Clarke, O.~Grumberg, and D.~Peled.
\newblock {\em Model Checking}.
\newblock MIT Press, Dec. 1999.

\bibitem{CONDON1992203}
Anne Condon.
\newblock The complexity of stochastic games.
\newblock {\em Information and Computation}, 96(2):203--224, 1992.
\newblock \href
  {http://dx.doi.org/https://doi.org/10.1016/0890-5401(92)90048-K}
  {\path{doi:https://doi.org/10.1016/0890-5401(92)90048-K}}.

\bibitem{daviaud2018pseudo}
Laure Daviaud, Martin Jurdzi{\'n}ski, and Ranko Lazi{\'c}.
\newblock A pseudo-quasi-polynomial algorithm for mean-payoff parity games.
\newblock In {\em Logic in Computer Science (LICS)}, pages 325--334, 2018.

\bibitem{de2001interface}
Luca De~Alfaro and Thomas~A Henzinger.
\newblock Interface automata.
\newblock {\em ACM SIGSOFT Software Engineering Notes}, 26(5):109--120, 2001.

\bibitem{dill1989trace}
David~L. Dill.
\newblock {\em Trace theory for automatic hierarchical verification of
  speed-independent circuits}, volume~24.
\newblock MIT press Cambridge, 1989.

\bibitem{Everest:2003}
Graham Everest, Alfred~Jacobus van~der Poorten, Igor Shparlinski, and Thomas
  Ward.
\newblock {\em Recurrence sequences}.
\newblock ACM, 2003.

\bibitem{Filar_Vrieze:book}
J.~Filar and K.~Vrieze.
\newblock {\em Competitive {Markov} Decision Processes}.
\newblock Springer, 1997.

\bibitem{gillette1957stochastic}
Dean Gillette.
\newblock Stochastic games with zero stop probabilities.
\newblock {\em Contributions to the Theory of Games}, 3:179--187, 1957.

\bibitem{GH2010}
Hugo Gimbert and Florian Horn.
\newblock Solving simple stochastic tail games.
\newblock In {\em ACM-SIAM Symposium on Discrete Algorithms (SODA)}, pages
  847--862, 2010.
\newblock URL: \url{http://epubs.siam.org/doi/abs/10.1137/1.9781611973075.69}.

\bibitem{GK2022}
Hugo Gimbert and Edon Kelmendi.
\newblock Submixing and shift-invariant stochastic games.
\newblock 2022.
\newblock URL: \url{https://arxiv.org/abs/1401.6575}.

\bibitem{Gimbert2011ComputingOS}
Hugo Gimbert, Youssouf Oualhadj, and Soumya Paul.
\newblock Computing optimal strategies for {Markov} decision processes with
  parity and positive-average conditions.
\newblock 2011.
\newblock URL: \url{https://hal.science/hal-00559173/en/}.

\bibitem{jurdzinski1998deciding}
Marcin Jurdzi{\'n}ski.
\newblock Deciding the winner in parity games is in {UP $\cap$ co-UP}.
\newblock {\em Information Processing Letters}, 68(3):119--124, 1998.

\bibitem{Maitra-Sudderth:2003}
A.~Maitra and W.~Sudderth.
\newblock Stochastic games with {Borel} payoffs.
\newblock In {\em Stochastic Games and Applications}, pages 367--373. Kluwer,
  Dordrecht, 2003.

\bibitem{M1998}
Donald~A. Martin.
\newblock The determinacy of {Blackwell} games.
\newblock {\em Journal of Symbolic Logic}, 63(4):1565–1581, 1998.

\bibitem{MSTW2017}
Richard Mayr, Sven Schewe, Patrick Totzke, and Dominik Wojtczak.
\newblock {MDPs with Energy-Parity Objectives}.
\newblock In {\em Logic in Computer Science (LICS)}. IEEE, 2017.
\newblock URL: \url{https://arxiv.org/abs/1701.02546}.

\bibitem{MSTW2021}
Richard Mayr, Sven Schewe, Patrick Totzke, and Dominik Wojtczak.
\newblock Simple stochastic games with almost-sure energy-parity objectives are
  in {NP} and {coNP}.
\newblock In {\em Proc.\ of Fossacs}, volume 12650 of {\em LNCS}, 2021.
\newblock Extended version on arXiv.
\newblock URL: \url{https://arxiv.org/abs/2101.06989}.

\bibitem{OW:2015}
Jo\"{e}l Ouaknine and James Worrell.
\newblock On linear recurrence sequences and loop termination.
\newblock {\em ACM SIGLOG News}, 2(2):4–13, 2015.

\bibitem{Piribauer:2021}
Jakob Piribauer.
\newblock {\em On non-classical stochastic shortest path problems}.
\newblock PhD thesis, Technische Universität Dresden, Germany, 2021.
\newblock URL:
  \url{https://nbn-resolving.org/urn:nbn:de:bsz:14-qucosa2-762812}.

\bibitem{Piribauer-Baier:2020}
Jakob Piribauer and Christel Baier.
\newblock {On Skolem-Hardness and Saturation Points in Markov Decision
  Processes}.
\newblock In Artur Czumaj, Anuj Dawar, and Emanuela Merelli, editors, {\em
  Proc.~of ICALP}, volume 168 of {\em LIPIcs}, pages 138:1--138:17, Dagstuhl,
  Germany, 2020. Schloss Dagstuhl--Leibniz-Zentrum f{\"u}r Informatik.
\newblock URL: \url{https://drops.dagstuhl.de/opus/volltexte/2020/12545}, \href
  {http://dx.doi.org/10.4230/LIPIcs.ICALP.2020.138}
  {\path{doi:10.4230/LIPIcs.ICALP.2020.138}}.

\bibitem{Piribauer-Baier:2023}
Jakob Piribauer and Christel Baier.
\newblock {Positivity-hardness results on Markov decision processes}.
\newblock 2023.
\newblock URL: \url{https://arxiv.org/abs/2302.13675v1}.

\bibitem{pnueli1989synthesis}
Amir Pnueli and Roni Rosner.
\newblock On the synthesis of a reactive module.
\newblock In {\em Annual Symposium on Principles of Programming Languages
  (POPL)}, pages 179--190, 1989.

\bibitem{ramadge1987supervisory}
Peter~J. Ramadge and W.~Murray Wonham.
\newblock Supervisory control of a class of discrete event processes.
\newblock {\em SIAM journal on control and optimization}, 25(1):206--230, 1987.

\bibitem{shapley1953stochastic}
Lloyd~S. Shapley.
\newblock Stochastic games.
\newblock {\em Proceedings of the national academy of sciences},
  39(10):1095--1100, 1953.

\bibitem{Zielonka:1998}
W.~Zielonka.
\newblock Infinite games on finitely coloured graphs with applications to
  automata on infinite trees.
\newblock {\em Theoretical Computer Science}, 200(1-2):135--183, 1998.

\end{thebibliography}

\newpage
\appendix
\section{Appendix for \Cref{sec:Gain}}\label{app:gain}

\begin{definition}\label{def:fix-fd}
Given an SSG $\game=\gametuple$ and a finite memory deterministic (FD)
strategy $\ostrat = \memstrattuple$ for Minimizer
let $\game_{\ostrat}$
be the maximizing
MDP with state space $\memconfset \x \states$
obtained by fixing Minimizer's
choices according to
$\ostrat$.
The transition rules $\movesto'$ in the derived MDP $\game_{\ostrat}$ are given as follows.
\begin{enumerate}
\item
  If $\s \in \zstates$
  for every $(\s,\s') \in \transition$, $\memconf \in \memconfset$,
  we have $(\memconf,\s)\, \movesto'\, (\memup(\memconf,(\s,\s')),\s')$,
  i.e., Maximizer determines the successor state and Minimizer updates its
  memory according to the observed transition. 
\item
  Similarly if $\s \in \rstates$
  for every $(\s,\s') \in \transition$, $\memconf \in \memconfset$
  we have
  $(\memconf,\s)\, \movesto'\, (\memup(\memconf,(\s,s')),\s')$
  and
  $\probp((\memconf,\s))((\memup(\memconf,(\s,s')),\s')) = \probp(\s)(\s')$,
  i.e., transition probabilities are inherited and Minimizer's memory is
  updated according to the observed transition.
\item
  If $\s \in \ostates$ then
  $(\memconf,\s)\, \movesto'\, (\memup(\memconf,(\s,\s')), \s')$ where $\s' = \memsuc(\memconf,\s)$,
  i.e., Minimizer chooses the successor state according to the
  strategy $\ostrat$ and updates its memory accordingly.
\end{enumerate}
The reward of each transition is the same as the reward of the 
transition in $\game$ from which it is derived. Similarly for the priorities
(aka coloring) of the states.
The ownership of the vertices $(\memconf, \s)$ in $\game_{\ostrat}$ is as follows.
If $\s \in \zstates$ then $(\memconf, \s)$ belongs to Maximizer.
If $\s \in \rstates$ then $(\memconf, \s)$ is also a chance vertex.
If $\s \in \ostates$ then $(\memconf, \s)$ also becomes a chance vertex (with
exactly one successor), since Minimizer's choice has been fixed.

In the dual case where a FD strategy $\zstrat$ for Maximizer is fixed,
we obtain a minimizing MDP $\game^{\zstrat}$. The construction is the same as
above, with the roles of Minimizer and Maximizer swapped.
\end{definition}

\section{Appendix for \Cref{sec:Computing N}}\label{app:mdp-approx}

\lemMDPN*

\begin{proof}
By \cref{lem:W0-W1-W2}(\cref{eq:nus+extra}) and
\cref{lem:exp_bound_term_value}, we have 
\[
  \valueof[\mdp][\Term(i)\,\cup\, \Oparity]{\s} - \valueof[\mdp][\Loss]{\s} \le  
\sup_{\zstrat}\Prob{\zstrat,\s}\lrc{\Term(i) \, \cap\, \complementof{\eventually\,W_1}} \le
\frac{c^i}{1-c}
\]
for all $i \ge h$ and $\s \in \states$.
To obtain a bound $N \ge h$ with $\frac{c^N}{1-c} \le \eps$, it suffices to choose
\[
  N \eqdef \max\lrc{h,\ceil{\log_c\lrc{\eps \cdot \lrc{1-c}}}}.
\]
We observe that
$\log_c\lrc{\eps \cdot \lrc{1-c}} = -\ln\lrc{\eps \cdot \lrc{1-c}} \cdot (-\ln(c)^{-1})$.\\
However, $-\ln(c) = -\ln(1-(1-c)) \ge (1-c)$.
Thus $\log_c\lrc{\eps \cdot \lrc{1-c}} \le
\ln\lrc{\frac{1}{\eps} \cdot \frac{1}{1-c}} \cdot \frac{1}{1-c}$.
For rewards in unary, by \cref{lem:exp_bound_term_value}, we have 
$1/(1-c) \in {\Ocompl}\left(\exp(\size{\mdp}^{\Ocompl(1)})\right)$ and
$h$ is only ${\Ocompl}(\exp\lrc{\size{\mdp}^{{\Ocompl}(1)}})$.
Thus $N \in {\Ocompl}\left(\exp(\size{\mdp}^{\Ocompl(1)})\cdot \log\lrc{1/\eps}\right)$.

Now consider the case where rewards are given in binary.
Following the proof of \cite[Lemma~3.9]{BBEK:IC2013},
the bounds are derived from the size of solutions of the constructed linear
program. While the MDPs in \cite{BBEK:IC2013} only consider unary rewards
from $\set*{-1,0,1}$,
one can extend it to the case where the rewards come from the set
$\set*{-R,\dots,0,\dots,R}$ in a natural way.
This affects the complexity of the above computed constants and thereby size of $N$.
More precisely, the proof of~\cref{lem:exp_bound_term_value}
can be split into three steps.
Firstly, given an MDP $\mdp$ construct a new ``rising''
MDP $\mdp'$. Then from this derived $\mdp'$,
construct a linear program.
From the solutions of constructed LP,
compute the required $c$ and $h$.
We evaluate the effect of having non-unary rewards in each of these steps.

When rewards are given in unary, the resulting $\mdp'$ has overall size
$\size{\mdp'} \leq 10 \size{\mdp}^4$. More exactly, $\size{\states'} \leq
10*\size{\states}^3*\lrc{\size{\states}+\size{\transition}}$ and similarly for
$\size{\transition'}$. When the rewards are given in binary, the construction results in an additional $R^2$ factor. So the resulting $\mdp'$ is pseudo-polynomially big when compared to $\mdp$ in our case.

The constructed LP (cf.~\cite[Fig.1]{BBEK:IC2013}) has $\states' + 2$ variables ($z_\s$ for each state, $x$ for the mean payoff and $\xi$ for converting the constraint $x > 0$ to $x \geq \xi $). Moreover all variables can be assumed non-negative. The number of constraints is bounded by $\transition' + 1$. Furthermore all the constants appearing in the constraints are either constants in the original MDP $\mdp$ or 1 or 0.

Finally, from an optimal solution of the LP $\tuple{z_\s,x,\xi}$ one can
compute $\exp\lrc{\frac{-x^2}{2\cdot(z_{\max}+x+R)^2}}$ and to get $c$, then
take a rational over-approximation and also take $h$ as $\ceil{z_{\max}}$ where $z_{\max} \eqdef \max_{\s \in \states'}z_\s - \min_{\s \in \states'}z_\s$. The only difference compared to the unary rewards case here is that the one step change of the submartingale is bounded by $z_{\max}+x+R$ instead of $z_{\max}+x+1$.

From the complexity point of view, both the construction of the LP and
the computation from its optimal solutions aren't affected by
changes in the rewards, i.e., the
previous bounds for $c$, $h$ and $N$ in terms of $\size{\mdp'}$ still hold.
In particular, $c \in
{\Ocompl}\left(\exp\lrc{1/2^{\size{\mdp'}^{{\Ocompl}(1)}}}\right)$, $h \in
{\Ocompl}(\exp\lrc{\size{\mdp'}^{{\Ocompl}(1)}})$ and thus
$N \in {\Ocompl}\lrc{\exp\lrc{\size{\mdp'}^{{\Ocompl}(1)}}\cdot
  \log\lrc{1/\eps}}$ by
\cite[Claim~6]{BBEK:IC2013}.

While previously, $\mdp'$ is only polynomially larger
than $\mdp$, introducing binary rewards blows up the
construction (cf.~\cite[Appendix A.2]{BBEK:IC2013}).
As a result we have that
$\size{\mdp'} \in {\Ocompl}\lrc{2^{\size{\mdp}^{{\Ocompl}(1)}}}$.
Therefore $N$ can be doubly exponential in the size of the original MDP
$\mdp$,
i.e.,
$N \in {\Ocompl}\left(\exp(\exp(\size{\mdp}^{\Ocompl(1)}))\cdot \log\lrc{1/\eps}\right)$.
\end{proof}

\section{Appendix for \Cref{sec: Unfolding till N}} \label{app:unfolding}
\begin{definition}[Definition of $\game\lrb{N}$] We present formally the definition of the game $\game\lrb{N}$, which unfolds the energy level in $\game$ till $N$
  $$\game\lrb{N} = \tuple{\states\lrb{N},\lrc{\zstates\lrb{N},\ostates\lrb{N},\rstates\lrb{N}},\transition\lrb{N}, \probp\lrb{N}}$$ where
  \begin{enumerate}
    \item $\states\lrb{N}\eqdef \states \times \set*{0,\ldots,N+R} \uplus \set*{\s_{\mathrm{win}},\s_{\mathrm{lose}}}$, the set of states is the tuple with the game state and energy level until $N+R$ as the maximum change in a single step is $R$ and since we are only interested in energy levels $\le N$, it suffices to consider till $N+R$.
    \item $\xstates\lrb{N}\eqdef \xstates \times \set*{1,\ldots, N}$, both players control their respective states until energy level $N$. Every state with energy $> N$ becomes a chance node. Consequently,
    \item $\rstates\lrb{N}\eqdef \rstates \times \set*{1,\ldots, N} \cup \states \times \set*{0,N+1,\dots,N+R} \cup \set*{\s_{\mathrm{win}},\s_{\mathrm{lose}}}$, since the Maximizer loses when the energy level becomes $\le 0$, we make these states as a chance vertex which go to a losing loop.
    \item $\transition\lrb{N}$, $\probp\lrb{N}$ 
    \begin{enumerate}
        \item For $0 < i \leq N$, $\lrc{\s,i} \movesto \lrc{\s',\max(0,j)}$ iff $\s \energymove{j-i} \s' \in \transition$, this is just simulating the transitions of the game until energy level $N$ and taking care of border cases. When energy drops below $0$, we move to level $0$ as there is no difference. When it shoots above $N$, it cannot go beyond $N+R$ and thus the transition is well defined.
        \item If $\s \in \rstates$ above, then the probability is carried over.
        \item $\lrc{\s,0} \movesto \s_{\mathrm{lose}}$ with probability 1.
        \item $\lrc{\s,N+k} \movesto \s_{\mathrm{win}}$ with probability $\valueof[\game][\en(N+k) \, \cap \, \Eparity]{\s}$ and with remaining probability moves to $\s_{\mathrm{lose}}$ for $1 \leq k \leq R$
        \item $\s_{\mathrm{lose}} \movesto \s_{\mathrm{lose}}$ with probability 1. Similarly for $\s_{\mathrm{win}}$.
    \end{enumerate}
\end{enumerate}
\end{definition}

\end{document}